\documentclass[review]{elsarticle}

\usepackage{hyperref}

\usepackage{amssymb}
\usepackage{amsmath}
\usepackage[graphicx]{realboxes}
\usepackage{diagbox}

\usepackage{tabularx}  
\usepackage{colortbl}
\usepackage{hhline}
\newcommand\LL[1]{\multicolumn{1}{|c}{#1}}
\newcommand{\ignore}[1]{}
\usepackage{booktabs}
\usepackage{multirow}
\usepackage{xcolor}

\newtheorem{theorem}{Theorem}

\newtheorem{corollary}[theorem]{Corollary}
\newdefinition{remark}{Remark}
\newproof{proof}{Proof}

\journal{Journal of Complexity}

\begin{document}

\begin{frontmatter}

\title{Minmax optimal list searching with $\log \log n$ average cost}


\author[address_1]{Ivo F. D. Oliveira}
\ead{ivodavid@gmail.com}

\author[address_2]{Ricardo H. C. Takahashi\corref{mycorrespondingauthor}}
\cortext[mycorrespondingauthor]{Corresponding author}
\ead{taka@mat.ufmg.br}

\address[address_1]{Institute of Science, Engineering and Tecnology of the Federal University of the Valleys of Jequitinhonha and Mucuri, Te\'{o}filo Otoni, Minas Gerais, Brazil }
\address[address_2]{Department of Mathematics of the Federal University of Minas Gerais, Belo Horizonte, Minas Gerais, Brazil}

\begin{abstract}
We find a searching method on ordered lists that surprisingly outperforms binary searching with respect to average query complexity while retaining minmax optimality. The method is shown to require $O(\log_2 \log_2 n)$ queries on average while never exceeding $\lceil\log_2 n\rceil$ queries in the worst case, i.e. the minmax bound of binary searching. Our average results assume a uniform distribution hypothesis similar to those of previous authors under which the expected query complexity of interpolation search of $O(\log_2 \log_2 n)$ is known to be optimal. Hence our method turns out to be optimal with respect to both minmax and average performance. We further provide robustness guarantees and perform several numerical experiments with both artificial and real data. Our results suggest that time savings range roughly from a constant factor of 10\% to 50\% to a logarithmic factor spanning orders of magnitude when different metrics are considered.
\end{abstract}


\end{frontmatter}

\section{Introduction}

Given a sorted list $\boldsymbol{v}$ and a target value $z$, the problem of searching sorted lists is typically stated as:
\begin{equation}\label{def.problem}\text{Find }k_*\text{ such that } v_{k_*} \leq z<v_{k_*+1};\end{equation} where $\boldsymbol{v}$ is  of size $n+1$ with entries in $[0,1]$ and $v_0 = 0$, $v_n=1$ and $z$ a scalar in $(0,1)$. This problem is ubiquitous in computer science with applications spanning several different fields of computer programming, engineering and mathematics. Variations of (\ref{def.problem}) include searching unbounded lists \cite{bentley},  tables \cite{knuth}, searching continuous functions for a zero \cite{oliveira}, as well as the construction of insertions and deletion procedures in canonical data-structures \cite{bentley2}. 

The standard approach to solve (\ref{def.problem}), commonly known as \textit{binary search} \cite{knuth},  begins with and updates upper and lower bounds $a = 0$ and $b = n$ for the location of the desired index $k_*$. At each step this is done by recursively probing the index $k_{1/2}$ which is defined by rounding
\begin{equation}\label{eq:def_binary_search}
    x_{1/2} \equiv\frac{a+b}{2}
\end{equation}
arbitrarily to the nearest integer, and, by comparing  $v_{k_{1/2}}$ and $z$ it updates $a$ and $b$ accordingly, i.e. if $v_{k_{1/2}}>z$ then $b$ is updated to $b=k_{1/2}$ if $v_{k_{1/2}}<z$ then $a$ is updated to $a=k_{1/2}$ and if $v_{k_{1/2}}=z$ then $a$ and $b$ are updated to $a=k_{1/2}$ and $b=k_{1/2}+1$. The algorithm terminates when the tolerance $\Delta\equiv b-a$ is equal to one. For convenience we display below the general structure of the binary searching algorithm as a while-loop; however, binary searching also admits for-loop formulations and other formulations that exploit computer architecture \cite{cannizzo,schlegel}  to improve computational speed. Here, $a$ and $b$ are initiated at $a=0, b=n$, and, in line (1)  $\tilde{k}$ is taken to be equal to $k_{1/2}$.\pagebreak
\\
{\setlength{\parskip}{0.05em}
\noindent\rule{12.7cm}{0.4pt}
\begin{center}
 \textbf{Algorithm 0: The Bracketing Algorithm}\\\end{center}
Input: $\boldsymbol{v}$ and $z$

while ($b-a>1$), do:

\hspace{10pt} (1). Choose $\tilde{k} \in (a,b)$ and evaluate $v_{\tilde{k}}$;

\hspace{10pt} (2). Update $a$ and $b$ according to the values of $v_a,v_b,v_{\tilde{k}}$ and $z$;\\
Output: $k_* = a$\\
\noindent\rule{12.7cm}{0.4pt}
}

The key feature of binary search is it's \textit{minmax optimality}. That is, it requires at most \begin{equation} \label{eq:minmax}
    N_{1/2} = \left\lceil\log_2 n \right\rceil
\end{equation} queries to locate $k_*$ while no other method can provide the same guarantee in less than $N_{1/2}$ queries. This property is specifically of interest when the computational cost of one query is known to be much higher than the computation of the search procedure itself. This assumption is often made implicitly in the literature, and for the sake of clarity, it is assumed henceforth. 

While binary search is optimal with respect to the worst case metric, \emph{interpolation search} turns out to be a more efficient alternative with respect to expected query complexity if a uniform distribution is assumed, see \cite{perl1,perl2,yao}. Interpolation search is a bracketing algorithm with $\tilde{k}$ in line (1) of Algorithm 0 defined as the linear interpolation between the points $(a,v_a)$ and $(b,v_b)$. More precisely $\tilde{k}$ is taken to be equal to $k_f$ where $k_f$ is defined as the integer closest to
\begin{equation}\label{eq:def_interp_search}
    x_f \equiv \frac{b\cdot (v_a-z) -a\cdot (v_b-z)}{v_a-v_b},
\end{equation}
that lies in between $x_f$ and $x_{1/2}$. The key feature of interpolation search is that if the entries of $\boldsymbol{v}$ and $z$ are sorted samples of a uniform distribution over $[0,1]$, then, interpolation search is optimal with respect to expected query complexity  \cite{yao} and the expected number of queries $\mathbb{E}(N)$ to solve (\ref{def.problem}) is  \begin{equation}\mathbb{E}(N) = O(\log_2 \log_2 n),\end{equation}  which considerably improves on the expected query complexity of binary search.

Although interpolation search enjoys an improved average performance, the improvement comes at the cost of requiring up to $n$ queries to terminate in the worst case. Furthermore, the guarantees on the expected query complexity of interpolation search do not hold if the distributional hypothesis is misspecified. Thus, choosing interpolation search over binary search may come at a high cost since interpolation search may also require a full series of $n$ queries to terminate on the average under different distributions. Is it possible to simultaneously enjoy the benefits of interpolation search with no costs on the minmax peformance of binary search? And furthermore, is it possible to enjoy such benefits without trading off performance under misspecified conditions?

In this paper we answer the above questions affirmatively. To answer the first question we begin by identifying the necessary and sufficient conditions for a searching method to be minmax optimal. Then, we pin-point one specific minmax method, which we name the  \emph{Interpolation, Truncation and Projection Method}, or simply the \emph{ITP Method}. We show that this method  attains the expected query complexity of $O(\log_2 \log_2 n)$ queries under similar assumptions as those required by interpolation search; and, it requires no more iterations than the upperbound of $\lceil\log_2 n\rceil$ of binary searching. Hence, it is both optimal with respect to minmax and average performance at no cost other than the computation of the method itself. To answer the second question, we find lower bounds on the average performance of binary searching under very broad distributional hypothesis and show that  the bisection method can never outperform the ITP method on the average performance by any significant margin. Hence, opting for the ITP method instead of binary searching comes at (almost) no cost even if the distribution is misspecified.

It is worth pointing out that our findings bear close resemblance with those of  \cite{oliveira} for the continuous version of Problem (\ref{def.problem}), i.e. searching the zero of a continuous function. However, despite the resemblance, our findings here are brand new and do not stem from previously known results. In fact, the methods for analysing the discrete searching problem in this paper are much more closely related to those developed in \cite{perl2} than those developed in the literature of numerical analysis. Perhaps more importantly, we believe that our findings here might be of more significance and repercussion than previous results due to the fundamental role that Problem (\ref{def.problem}) and Algorithm 0 plays in the field of computer science, serving as a basis for a much of algorithmic theory and practice.

\paragraph{Paper Outline} The following section, entitled \emph{Main Results},  is divided into three parts. The section begins by characterizing necessary and sufficient conditions for Algorithm $0$ to be minmax optimal, putting forward results analogous to Theorem 2.1 of \cite{oliveira} which were previously unknown in the discrete case. Then, in Subsection \ref{sec:ITP} we describe our main contribution: the ITP method for searching sorted lists with minmax and expected query complexity results in Theorem \ref{the:ITP_charaterization}. These results show that under mild conditions the ITP method can attain an expected query complexity of the same order of interpolation search while retaining the minmax optimality of binary search. In Subsection \ref{sec:robustness}, we calculate lower bounds on the expected query complexity of binary searching under very broad distributional assumptions, and as a consequence we find that our method cannot be outperformed on the average by binary search by more than one or two iterations. Thus, we provide brand new robustness guarantees that cannot be provided by interpolation search. In Section \ref{sec:experiments} we perform extensive experiments on both artificial and real data from which we find that the expected query complexity of the ITP method can be orders of magnitude lower than interpolation search and binary search alike. Finally, in Section \ref{sec:discussions} we summarize and discuss the relevance of our findings and point out applications and future directions of research. 

\section{Main Results}\label{sec:main_results}
Given a sorted list $\boldsymbol{v}$ and a target value $z$, at each iteration $j$ of Algorithm $0$ define $\Delta_j$ as $\Delta_j \equiv b_j-a_j$ and $x_{1/2} \equiv (a_j+b_j)/2$. Then,
\begin{theorem} \label{the:global_minmax}
Algorithm $0$ requires at most $N_{1/2}$ iterations to terminate if and only if at each iteration $j$ we have \begin{equation}\label{eq:minmax_range}|\tilde{k}_j-x_{1/2}|\leq 2^{N_{1/2}-j-1}-\tfrac{1}{2}\Delta_j.\end{equation}
\end{theorem}
\begin{proof}
Given any instance of (\ref{def.problem}),  we may calculate the maximum number of iterations $N_{1/2}$ required by any minmax strategy using equation (\ref{eq:minmax}).  After the first iteration one is left with $N_{1/2}-1$ queries and so, from (\ref{eq:minmax}) we have that $\lceil \log_2 (b_1-a_1) \rceil$ must be at most $N_{1/2}-1$, thus $b_1-a_1$ can be at most $2^{N_{1/2}-1}$. Combining this with the fact that $b_1-a_1$ is less than or equal to $b_0-\tilde{k}_0$ and $\tilde{k}_0-a_0$ it follows that as long as $\tilde{k}_0$ is chosen in such a way that both $b_0-\tilde{k}_0$ and $\tilde{k}_0-a_0$ are less than or equal to $2^{N_{1/2}-1}$, then, from that step onward, Algorithm $0$ can still guarantee termination in $N_{1/2}$ iterations. Requiring that  both $b_0-\tilde{k}_0$ and $\tilde{k}_0-a_0$  be less than or equal to $2^{N_{1/2}-1}$ is equivalent to enforcing $|\tilde{k}_0 - x_{1/2}|\leq 2^{N_{1/2}-1} - \tfrac{1}{2}\Delta_0$. This proves Theorem \ref{the:global_minmax} for iteration $j=0$.

For higher values of $j$ the reasoning is very similar. After steps $0,1,...,j-1$, Algorithm $0$ is left with $N_{1/2}-j$ iterations. Thus, on step $j$, as long as $b_{j+1}-a_{j+1}$ is less than or equal to  $2^{N_{1/2}-j}$, Algorithm $0$ can guarantee termination in at most $N_{1/2}-j$ iterations. Thus, we find similarly that $b_{j+1}-a_{j+1} \leq 2^{N_{1/2}-j}$ is guaranteed when $|\tilde{k}_j - x_{1/2}|\leq 2^{N_{1/2}-j-1} - \tfrac{1}{2}\Delta_j$, and, this completes the proof.
\end{proof}

Theorem \ref{the:global_minmax} identifies the class of strategies that, similar to binary searching, enjoy minmax optimality.  In most situations the set of strategies that satisfy the conditions of Theorem \ref{the:global_minmax} can be quite large. However, when $n$ is equal to $2^T$ for some $T\in\mathbb{N}$ then  we will find that  $2^{N_{1/2}-j-1}-\tfrac{1}{2}\Delta_j$ must be null for $j=0,1,2...$ and thus the class naturally reduces to binary searching. In every other situation  $\tilde{k}$ may be chosen by means of interpolation, randomization or any other technique as long as the distance of $k_j$ to $x_{1/2}$ remains within the ranges established by Theorem \ref{the:global_minmax}. Figure \ref{fig:classes} depicts  two search trees with $n=17$ that have depth $N_{1/2}=5$. Both of these trees have the same minimal depth of binary searching, however they do not subdivide the nodes in half in each query as binary searching would.   Given that $n =17 $ is not a power of $2$ then several such trees with depth $N_{1/2}=5$ exist. 

\begin{figure}
	\centering
	\includegraphics[width=100mm]{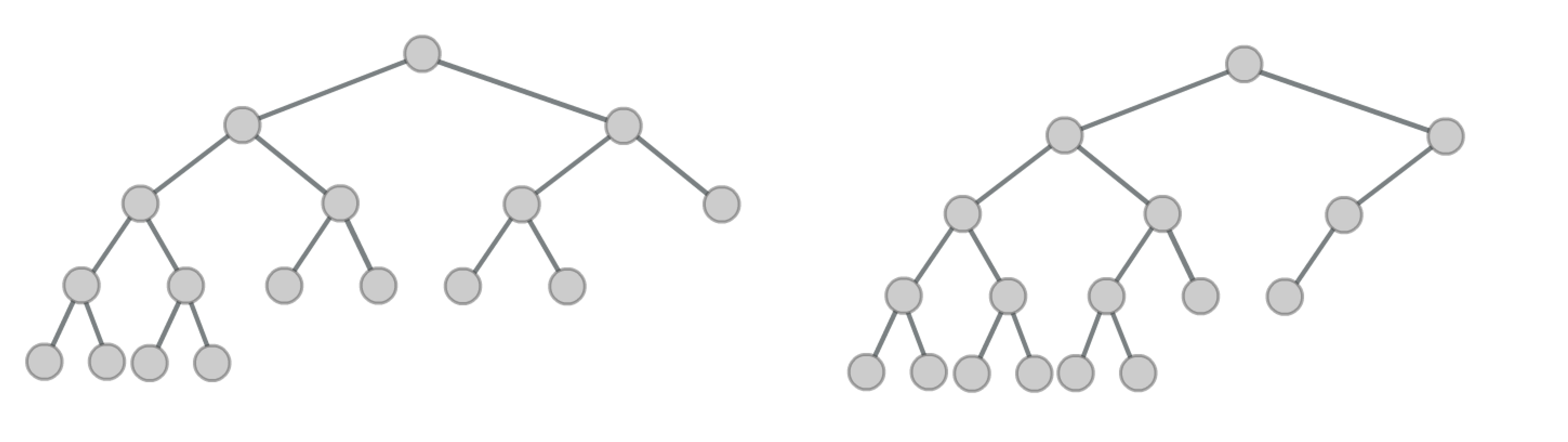}
	\caption{\label{fig:classes} Two searching trees with $17$ nodes and minimal depth of $5$, none of which correspond to binary searching. }
\end{figure}

Before we proceed in displaying our main algorithm we point out two variations of Theorem \ref{the:global_minmax} that may be of interest in different circumstances, one less conservative and one more conservative.  Both variations are motivated by the fact that minmax optimality alone does not avoid certain types of inefficiencies. The first of type arises from the fact that (\ref{eq:minmax_range}) may, at times, be too restrictive and degenerate to bisection steps too early in a run. This may happen if (\ref{eq:minmax_range}) is initiated too small, or, if Algorithm $0$ unluckily makes too many ``bad guesses'' with $\tilde{x}$  producing $\Delta_{j+1}>\Delta_{j}/2$ for several iterations. To avoid this and produce a variation of (\ref{eq:minmax_range}) that is more ``forgiving'' of bad iterations one may upperbound the maximum number of iterations by $N_{\max}$ instead of $N_{1/2}$ with $N_{\max}\geq N_{1/2}+1$. This is attained if and only if in each iteration $j$ we have
\begin{equation}
    \label{eq:almost_minmax}
|\tilde{k}_j-x_{1/2}|\leq 2^{N_{\max}-j-1}-\tfrac{1}{2} \Delta_j.\end{equation}

The second type of inefficiency that may be present is of opposite nature. Minmax optimality may allow for too much freedom. For example, it is possible that Algorithm $0$ after a few iterations reduces $\Delta$ to a sufficiently small size that it could be tackled with a few binary steps. However, minmax optimality allows Algorithm $0$ to waste ``spare iterations'' produced in the beginning of the run. One way to avoid this is to  require that after each iteration, the new subproblem with $b_j-a_j=\Delta_j$ would take no more iterations than binary search would, i.e. that at most $\left \lceil \log_2 \Delta_j \right \rceil$ queries would be used from  step $j$ onward.  This is obtained by enforcing 
\begin{equation}
    \label{eq:local_minmax}
|\tilde{k}_j-x_{1/2}|\leq 2^{\lceil \log_2 \Delta_j\rceil -1}-\tfrac{1}{2} \Delta_j\end{equation}
in every iteration instead of equation (\ref{eq:minmax_range}).

All three versions of (\ref{eq:minmax_range}) may be of interest to software development. If problem (\ref{def.problem}) is generated by a known distribution that allows for the construction of reliable estimators  for the location of $k_*$, as exemplified in the next section, then perhaps the original form (\ref{eq:minmax_range}) might be chosen. If (\ref{eq:minmax_range}) is too small, then the relaxation in (\ref{eq:almost_minmax}) might be an appropriate alternative. In fact allowing for as little as one additional iteration with $N_{\max} = N_{1/2}+1$ we have that equation (\ref{eq:almost_minmax}) will encompass the entirety of $(a,b)$ in the first iterations. Finally, equation (\ref{eq:local_minmax}) might be preferred if  the underlying distribution does not allow for the construction of reliable estimators for the location of $k_*$, or, if the underlying distribution is unknown. In any case, the classes of methods here identified by (\ref{eq:minmax_range}) to (\ref{eq:local_minmax}) offer a rich collection of alternatives to traditional binary searching that simultaneously retain minmax optimality and allow for enough freedom to incorporate interpolation and/or randomized strategies. In the following subsection we will see that even the ``unaltered'' minmax optimality condition in (\ref{eq:minmax_range}) can allow for an improved average performance under standard uniform distribution hypothesis.

\subsection{The ITP Method}\label{sec:ITP}
Let $\kappa_1\in \mathbb{R}_+$ and $\kappa_2\in(\tfrac{1}{2},1)$ be two user provided constants\footnote{Notice that  $\kappa_2$ is defined here to be between $1/2$ and $1$, whereas in \cite{oliveira} it is defined to be between $1$ and $(1+\sqrt{5})/2$. This difference is key and arises from the fact that in continuous settings one is typically interested in vanishing residuals $\Delta$ that are less than or equal to $1$, whereas in discrete scenarios $\Delta$ is always greater than $1$.}. Now define $\sigma$ and $\delta$ as
\begin{equation} \label{eq:sigma_delta}
     \sigma \equiv \text{sign}(x_{1/2}-x_f) \ \ \ \ \ \ \text{ and } \ \ \ \ \ \  \delta \equiv\kappa_1|b-a|^{\kappa_2},
\end{equation} where $x_{1/2}$ and $x_{f}$ are as in (\ref{eq:def_binary_search}) and (\ref{eq:def_interp_search}) respectively. Also, define $x_t$ as
\begin{equation} \label{eq:improved_falsi}
    x_t \equiv x_f + \sigma\delta
\end{equation}
if  $\delta \leq |x_{1/2}-x_f|$ and $x_t = x_{1/2}$ otherwise. Now  define the minmax radius $r_k$ and interval $\mathcal{I}_k$ as
\begin{equation}
r_j \equiv 2^{N_{1/2}-j-1}-\tfrac{b_k-a_k}{2} \ \ \ \ \ \text{ and } \ \ \ \ \
\mathcal{I}_j \equiv \left[ x_{1/2}-r_j \, , \, x_{1/2}+r_j \right]
\label{eq:minmax_range2}
\end{equation}
Now, in each step $j$ define $x_{ITP}$ as the projection of $x_t$ onto $\mathcal{I}_j$, i.e.   \begin{equation}x_{ITP} \equiv \left\{ \begin{array}{ll}
   x_t  &  \text{ if }  |x_t-x_{1/2}|\leq r_j; \\
    x_{1/2} -\sigma r_j   & \text{ otherwise.}
\end{array} \right. \end{equation}
The ITP method then takes $\tilde{k}$ to be equal to $k_{ITP}$ defined as the closest integer to $x_{ITP}$ that lies between $x_{ITP}$ and $x_{1/2}$.

In the following theorem we will assume that 
 $\boldsymbol{v}$ is constructed by sorting $n$ independent samples of a uniformly distributed variable in $[0,1]$. And, that the minmax interval $\mathcal{I}_0$ around $x_{1/2}$ in the first iteration $j=0$ is ``not too small'', i.e. that $2r_0/n$ is not much smaller than one. This avoids the collapsing of  $\mathcal{I}_k$ to $x_{1/2}$ (in which case the ITP method behaves identical to binary searching), and also, as shown in the proof of Theorem \ref{the:ITP_charaterization}, in combination with the other conditions it guarantees that with high probability a steady state condition with super-linear convergence can be reached within a few iterations.

\begin{theorem}\label{the:ITP_charaterization}
If $n$ is sufficiently large and  $\mathcal{I}_0$ is not too small,  the number of iterations $N$ for Algorithm $0$ to terminate satisfies 
 \begin{equation}
     N\leq N_{1/2} \ \ \ \ \ \ \ \text{ and } \ \ \ \ \ \ \ \mathbb{E}(N) \leq \kappa_3\log_2 \log_2 n
 \end{equation}
 for some constant $\kappa_3\in \mathbb{R}$ that depends on $\kappa_1$ and $\kappa_2$ but not on $n$.
\end{theorem}
\begin{proof}
The structure of the proof is as follows: We begin by analysing Algorithm $0$ for $\tilde{k} = k_t$ where $k_t$ is the closest integer to $x_t$ that lies between $x_t$ and $x_{1/2}$, i.e. without the projection onto $\mathcal{I}_j$. We will see that for sufficiently large $n$ we have that $k_t$ produces an expected query complexity of the order of $\log_2 \log_2 n$. Then, we include the projection step and verify that, if $\mathcal{I}_0$ is not too small, then with high probability the minmax range $\mathcal{I}_j$ will at least double in each iteration and in a few iterations the full interval $(a,b)$ will be encompassed by $\mathcal{I}_j$. After that point, Algorithm $0$ will behave as if there were no projection step, and thus, the same expected query complexity of $k_t$ is attained.

Before we proceed with the proof, we point out that no attempt is made here to obtain the tightest bounds nor to optimize our choice of $\kappa_1,\kappa_2$ or any other meta-parameter. Instead, whenever possible we opted for the simplest and shortest path to obtain our results, and, overall aim for a proof that is accessible to a university level 
advanced algorithms course.

In order to calculate the expected query complexity of Algorithm $0$ implemented with $\tilde{k} = k_t$ we first calculate  the expected number of iterations $N_{\delta}\in\mathbb{R}$ that Algorithm $0$ requires to reduce an interval of length $\Delta$ to a new interval with length less than or equal to $8\delta =8\kappa_1\Delta^{\kappa_2}$. For this we will use a few facts. First notice that the distance between $k_t$ and $k_*$ can be upper-bounded by:
\begin{equation} \label{eq:errors}
|k_t-k_*|\leq |x_f-k_*|+ \kappa_1\Delta^{\kappa_2}+1.    
\end{equation}
We refer to the first term $|x_f-k_*|$ as the estimation error, the second term $\kappa_1\Delta^{\kappa_2}$ as the truncation error, and the third term ``+1'' is the round-off error.

We say that an iteration $j$ is \emph{successful} when $\Delta_{j+1}\leq \Delta_j/2$, and, it is \emph{unsuccessful} when $\Delta_{j+1}> \Delta_j/2$.   Notice that if $\boldsymbol{v}$ is built by sorting $n$ independent samples from a uniform distribution over $[0,1]$, then the probability of an iteration $j$ with $\tilde{k} = k_t$ being successful is equal to the probability that $k_t$ is between $k_*$ and $x_{1/2}$. Without loss of generality we may assume that $k_t\leq x_{1/2}$, and thus the probability of a successful iteration is equal to the probability that $k_*\leq k_t$. Now the index $k_*$ is equal to the number of entries $v_j$ of $\boldsymbol{v}$ that satisfy  $v_j\leq z$, and since each entry is sampled from a uniform distribution over $[0,1]$, then, in problem (\ref{def.problem})  the variable $k_*$ follows a binomial distribution with expected value of $\mu = n\cdot z$ and with variance $\varsigma^2=n z(1-z)\leq n/4$. During a run of Algorithm 0, given all the data collected up to iteration j, by using scaling arguments we find that the conditional distribution of  $k_*$ will also follow a translated binomial between $a_j$ and $b_j$, with mean $\mu = x_f$ and with variance  $\varsigma^2 \leq \Delta/4$. Thus, from (\ref{eq:errors}) and Chebyshev's inequality $\mathbb{P}(|y-\mu|\geq t)\leq \varsigma^2/t^2$  we find that:
\begin{equation}\label{eq:prob_unsuccessfull_iteration}
    \mathbb{P}(\text{unsuccessful iteration}) \ \ \leq \ \ \mathbb{P}(|k_*-x_f|\geq\delta) \ \  \leq \ \ \frac{1}{4\kappa_1^2\Delta^{2\kappa_2-1}}.
\end{equation}
And thus, since $\kappa_2>1/2$, the probability of an unsuccessful iteration vanishes with larger values of $\Delta$. We will denote by $P_s$ and $P_u$ the probabilities of successful and unsuccessful iterations respectively.

 Now,  from (\ref{eq:prob_unsuccessfull_iteration}) we have that for  large values of $\Delta$ the estimation error is smaller than the truncation error  with high probability. The same is true of the round-off error. Thus we deduce that 
\begin{equation} \label{eq:kt_error}
    |k_t-k_*|\leq 3\kappa_1\Delta^{\kappa_2}
\end{equation}
with high probability for large values of $\Delta$.

Now let us analyse two different scenarios: (i) when $k_*$ is near extremity $a$ or $b$; and (ii) when it is somewhere in the middle. Or, formally: (i) when $k_*-a\leq \kappa_1 \Delta^{\kappa_2}$ or $b-k_*\leq \kappa_1 \Delta^{\kappa_2}$; (ii)  every other case. It is easy to see that in case (i), with high probability, one successful iteration will suffice to reduce $\Delta$ to less than or equal to $4\kappa_1\Delta^{\kappa_2}$. This is a direct consequence of equation (\ref{eq:kt_error}) and (\ref{eq:prob_unsuccessfull_iteration}). Similarly, notice that case (ii) after one iteration  will produce  $k_*-a\leq 4\kappa_1 \Delta^{\kappa_2}$ or $b-k_*\leq 4\kappa_1 \Delta^{\kappa_2}$ with high probability. Thus, after two successful iterations case (ii) will reduce  $\Delta$ to less than or equal to  $7\kappa_1\Delta^{\kappa_2}$. Hence, with high probability, it suffices to obtain two successful iterations in order to reduce $\Delta$ to less than or equal to $8\kappa_1\Delta^{\kappa_2}$, and, the expected number of iterations required to obtain two successes is given by
$$ \mathbb{E}(\text{\# iterations to obtain two successes}) = (2\cdot 1 ) P_s^2+(3\cdot2) P_uP_s^2+(4\cdot 3) P_u^2P_s^2... $$
which, by using the relation $P_s+P_u = 1$ simplifies to 
$$\mathbb{E}(\text{\# iterations to obtain two successes}) = 2/P_s \xrightarrow{\Delta \to \infty} 2.$$
Thus we find that  $N_{\delta}$ approaches 2 as $\Delta$ goes to infinity, and hence for $\Delta$ greater than or equal to some constant $\kappa_4$ (that depends on $\kappa_1$ and $\kappa_2$ alone) we have $N_{\delta}$ less than or equal to $3$. This implies that for large $\Delta$ we have:
\begin{equation}\label{eq:recursive_relation}
    \mathbb{E}(N|\Delta) \leq 3+\mathbb{E}(N|8 \kappa_1\Delta^{\kappa_2}),
\end{equation}
where $\mathbb{E}(N|Z)$ is the expected number of iterations given that there are $Z$ elements in $[a,b]$. Thus applying (\ref{eq:recursive_relation}) recursively we find that $$\mathbb{E}(N|\Delta) \leq 3+3+\mathbb{E}(N| 8\kappa_1(8\kappa_1\Delta^{\kappa_2})^{\kappa_2})$$ and repeating this process $m$ times we find $$ \mathbb{E}(N|\Delta)\leq m\times3+\mathbb{E}(N| (8\kappa_1)^{\frac{\kappa_2^m-1}{\kappa_2-1}}\Delta^{\kappa_2^m}).$$ Thus, the value of $m$ such that $(8\kappa_1)^{\frac{\kappa_2^m-1}{\kappa_2-1}}\Delta^{\kappa_2^m}$ is less than a $\kappa_4$ will give us the expected query complexity of Algorithm $0$ implemented with $\tilde{k}=k_t$. This, of course, reduces to  $$m\leq C_1 + C_2\log_2\log_2\Delta $$
for some $C_1$ and $C_2$ that depend on $\kappa_1$ and on $\kappa_2$ but not on $\Delta$. This completes the deduction of the expected query complexity of Algorithm $0$ implemented with $k_t$. 

What is left now is to verify the effect of the projection step on the behaviour of Algorithm $0$. We start by pointing out that for high values of $\Delta$, due to (\ref{eq:prob_unsuccessfull_iteration}), Algorithm $0$ implemented with $k_t$ generates successful iterations with high probability.  The same is true for the projection of $k_t$ onto $\mathcal{I}_j$, since,  if $k_t$ lies between $k_*$ and $x_{1/2}$ then so will the projection of $k_t$ onto $\mathcal{I}_j$. Thus, with high probability we are left with the smallest amongst the intervals $(a,\tilde{k})$ and $(\tilde{k},b)$ after each iteration. This implies that $\Delta_{j+1}\leq \Delta_j/2$ and, with a little algebra, we can show that the fraction of the interval $(a_j,b_j)$ covered by $\mathcal{I}_j$, which we will denote by $F(j)$, increases from iteration $j$ to iteration $j+1$, i.e. that $$F(j+1) \equiv \frac{2(2^{N_{1/2}-(j+1)-1}-\Delta_{j+1}/2)}{\Delta_{j+1}} \geq  \frac{2(2^{N_{1/2}-j-1}-\Delta_{j}/2)}{\Delta_{j}} \equiv F(j).$$
Furthermore, if $\tilde{k}$ is not equal to $k_t$, then it is the projection onto the minmax disk. Thus, (ignoring rounding effects) we find that $\tilde{k}_j = x_{1/2}\pm (2^{N_{1/2}-j-1}-\Delta_j/2)$. In which case since $\Delta_{j+1}\leq \Delta_j/2$ we must have that:
$$\Delta_{j+1} = \Delta_j - 2^{N_{1/2}-j-1};$$
from which we derive that the fraction of the interval $(a_{j+1},b_{j+1})$ covered by $\mathcal{I}_j$ is given by
$$F(j+1)\equiv$$ $$\frac{2(2^{N_{1/2}-(j+1)-1}-\Delta_{j+1}/2)}{\Delta_{j+1}} = \frac{2(2^{N_{1/2}-j-1}-\Delta_{j}/2)}{\Delta_{j}}\times \frac{1}{1- 2^{N_{1/2}-j-1}/\Delta_j},$$
and thus $F(j+1)> 2F(j)$ since $2^{N_{1/2}-j-1}/\Delta_j$ is greater than $1/2$ given that $\Delta_j$ is not a power of two. Thus, with high probability the fraction of the interval $(a_{j+1},b_{j+1})$ covered by $\mathcal{I}_j$ must at least double in each iteration if $\tilde{k}\neq k_t$. Hence, if the fraction of the interval $(a_{0},b_{0})$ covered by $\mathcal{I}_0$ is not too small, it will take only a few iterations until $\tilde{k}$ can assume any value within $(a,b)$, and will thus coincide with $k_t$ from that iteration onward. \end{proof}

Theorem \ref{the:ITP_charaterization} shows that the ITP method is both minmax optimal and can attain as low as $\log_2\log_2 n$ expected query complexity given that $\mathcal{I}_0$ is not too small. This last condition, as mentioned earlier, can be dropped if minmax optimality is relaxed. In fact it suffices to allow for just one iteration more than $N_{1/2}$ and the ``not too small'' condition is satisfied. Also, it is worth mentioning that one may calculate the expected number of ``gained'' iterations per query and find that it is greater than or equal to one for sufficiently large $n$. Thus, even though $\mathcal{I}_j$ can collapse into binary searching after a few rounds of unsuccessful iterations, this will only happen with low probability since in an average run, the ITP method it will typically accumulate ``spare iterations'' and can afford a few misses quite early in the run.  Also, from (\ref{eq:prob_unsuccessfull_iteration}) we may deduce that the first iterations have the highest probability of being successful since these have the biggest intervals $\Delta$, and hence, the iterations in which Algorithm $0$ has less ``spare iterations'' are the ones less likely to blunder and produce unsuccessfull/wastefull iterations. By the time it has narrowed down the search to smaller intervals, several ``spare'' iterations will be available, and thus it will take many more unsuccessfully iterations for $\mathcal{I}_j$ to degenerate to binary searching.

\subsection{Robustness and Limits}\label{sec:robustness}

It is well known that the expected query complexity of interpolation search is of the order of $\log_2 \log_2 n$ and binary search is of the order of $\log_2 n$ under the uniform distribution assumption, i.e. interpolation search considerably outperforms binary search under the standard hypothesis. However, it is also well known that if the distribution is misspecified, then,  the expected query complexity of interpolation search can reach up to $n$ queries while binary search remains upper bounded by $\lceil \log_2 n\rceil$, i.e. interpolation search is considerably outperformed by binary search under misspecified conditions. In this section we will verify whether the ITP method  suffers from the same drawback or whether it is robust to such changes, i.e. can the ITP method be outperformed by binary searching with respect to average performance\footnote{In the continuous setting this question was answered in Corollary 2.2 of \cite{oliveira}. There, since the bisection method has a fixed expected query complexity of $N_{1/2}$ for any continuous distribution, the worst-case guarantees of the ITP method already ensure that the expected query complexity of the ITP method cannot be outperformed by the bisection method. However, unlike the continuous setting, the expected query complexity of binary searching over lists does depend on  the underlying distribution.}?

We will answer this question by analysing two large classes of distributions over instances of (\ref{def.problem}). The first class , which we will denote by $\mathcal{C}_1$, encompasses all distributions over instances of (\ref{def.problem}) that produce $z \neq v_{k_*}$ with no restriction on how $z$ and $\boldsymbol{v}$ are generated. The second class, denoted $\mathcal{C}_2$, encompasses a large collection of distributions over instances of (\ref{def.problem}) that do produce $z = v_{k_*}$. In particular, the second class includes any distribution that does not limit $k_*$ nor favours any particular $j$ between $1$ and $n$, i.e. it assumes only that $k_*$ that can assume any value from $1$ to $n$ with a uniform probability. 

If Problem (\ref{def.problem}) is generated by a distribution from class $\mathcal{C}_1$, then since the distribution does not produce $z = v_{k_*}$, binary search must require at least $N_{1/2}-1$ iterations to terminate. To see this, first notice that $\Delta_{j}\geq \tfrac{1}{2}(\Delta_{j-1}-1)$, for any value of $\Delta_j= b_j-a_j$ (whether odd or even). By recursion, we find  that $\Delta_{j} \geq \tfrac{1}{2^j} \Delta_0 - \tfrac{1}{2}- \tfrac{1}{4}-...- \tfrac{1}{2^j}$ which in turn is greater than $\tfrac{1}{2^j} \Delta_0 - 1$. Hence, in order for $\Delta$ to be less than or equal to $1$  the number of iterations $N$ must satisfy: $ N \geq \lceil \log_2 \Delta_0\rceil -1$. Thus
\begin{corollary}\label{cor:unlucky_distributions}
If the distribution over instances of (\ref{def.problem}) is such that $z \neq v_{k_*}$ then the expected query complexity of binary searching satisfies:
\begin{equation}
    \mathbb{E}(N)\geq N_{1/2}-1.
\end{equation}
\end{corollary}

The second class of distributions does allow for Problem (\ref{def.problem}) to admit a solution with $z = v_{k_*}$. The class $\mathcal{C}_2$ assumes nothing on how $\boldsymbol{v}$ or $z$ is constructed other than the fact that the solution $k_*$ can assume any value within the range from 1 to n with a uniform probability\footnote{This second constraint is added since otherwise it is easy to construct distributions that can arbitrarily lower the expected query complexity of virtually any method. Taking binary search as an example, if the distribution trivially produces $v_{k_{1/2}}= z$ then the expected query complexity can be as low as one iteration. Thus to exclude trivial cases and arbitrary biases we assume that $k_*$ is equally likely to assume any value between $1$ to $n$.}. In this second case it is useful to consider the graph constructed by placing the first index visited on the root, and, successively branching left with the indices probed in case of $z< v_k$ and branching right when $v_k< z$. Figure \ref{fig:tree} illustrates one such construction. \begin{figure}
	\centering
	\includegraphics[width=70mm]{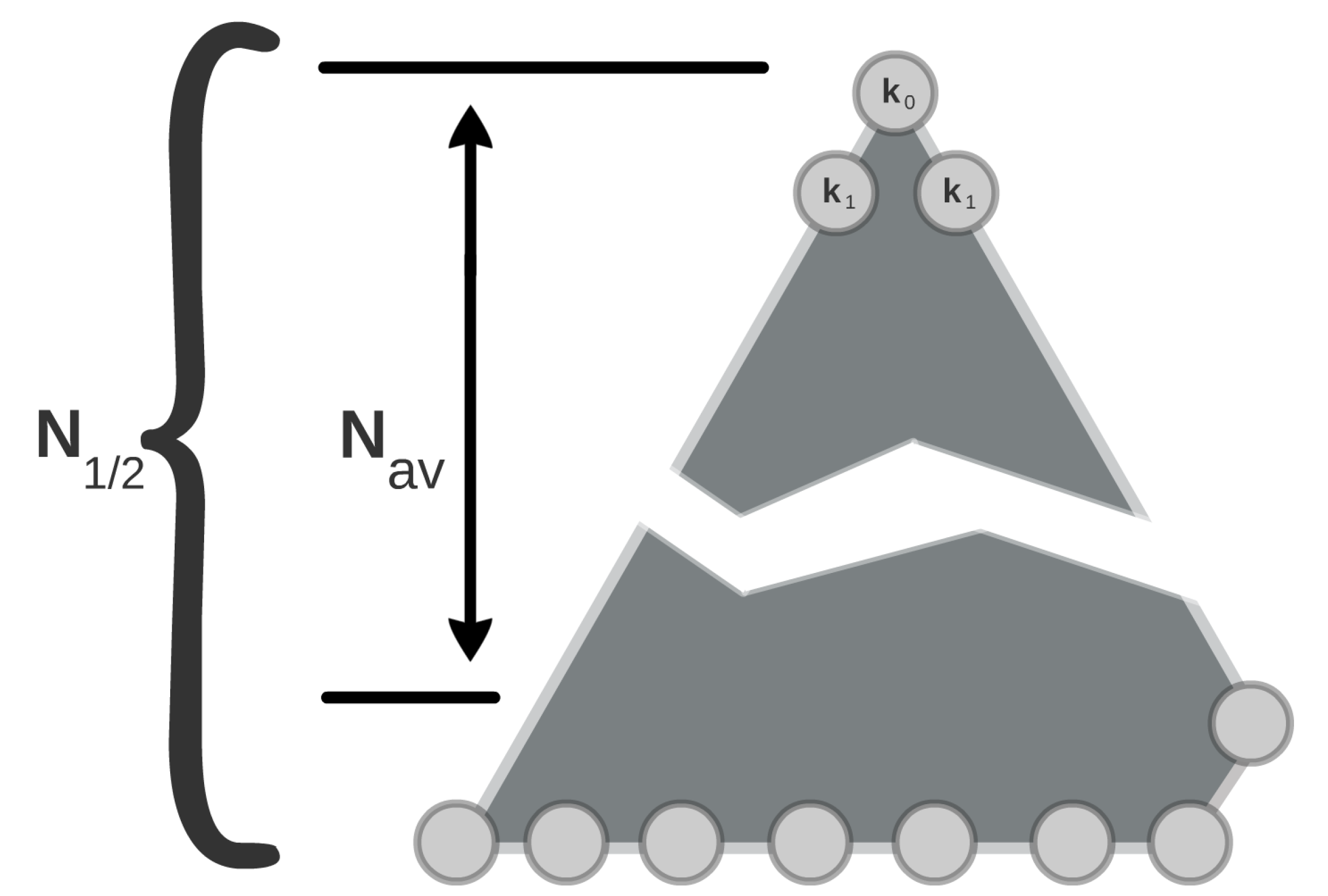}
	\caption{\label{fig:tree} The binary search tree associated with Algorithm $0$. Each node of the tree represents an index $k$ of vector $\boldsymbol{v}$ visited by Algorithm $0$, the height of the tree represents the worst-case complexity of the searching strategy and the average depth of the tree $N_{av}$ is represents the expected query complexity of the searching strategy. }
\end{figure} The depth of the resulting tree measures the maximum number of iterations required for Algorithm $0$ to terminate, and, the average depth of the graph measures the average number of iterations. We will denote the average depth by $N_{av}$ and we decompose $n$ into two factors as $n = 2^{N_{1/2}-1}+q$ for some $q\in\{1,..,2^{N_{1/2}-1}\}$. This way we find that 
\begin{corollary}\label{cor:lucky_distributions} If the distribution over instances of (\ref{def.problem}) is such that $z=v_{k_*}$ and $k_*$ is equally likely to assume any value between $1$ to $n$ then the expected query complexity of binary searching is equal to $N_{av} = N_{1/2}-1 - \delta$ where $\delta = \frac{n-N_{1/2}-2q}{n-1}$, and satisfies
\begin{equation} \label{eq:Nav}
   \mathbb{E}(N) \geq  N_{1/2}-2.
\end{equation}
\end{corollary}

Corollary \ref{cor:lucky_distributions} is well known and it's proof is thus omitted for simplicity \footnote{For completeness sake we point to Prof. PhD Steven Pigeon's proof an analysis of Corollary \ref{cor:lucky_distributions} in \emph{Average node depth in a Full Tree} that can be found in \url{https://hbfs.wordpress.com/2013/05/14/average-node-depth-in-a-full-tree/}, published in 2013.}.  Now combining the above corollaries with the fact that the ITP method requires no more than $N_{1/2}$ iterations to terminate we find that for the classes of distributions in $\mathcal{C}_1$ and $\mathcal{C}_2$ described above:
\begin{theorem} \label{the:ITP_robustness}
The expected query complexity of binary searching can outperform that of the ITP method by at most two iterations.
\end{theorem}

Hence, unlike interpolation search, even under very broad misspecified conditions the ITP method cannot be outperformed by binary searching by any significant margin. Thus, Theorems \ref{the:ITP_charaterization} and \ref{the:ITP_robustness} combined show that by choosing the ITP method over binary searching, not only will Algorithm $0$ enjoy the benefits associated with interpolation search  (the $\log_2 \log_2 n$ complexity over the uniform distribution assumption), but it will also not suffer the drawbacks associated with interpolation search (being outperformed by binary searching under misspecified conditions).

\section{Experimental Results}\label{sec:experiments}
In this section we empirically test the ITP strategy on three experiments and compare it with traditional binary searching and interpolation search. In the first experiment we test the minmax ITP method with varying values of $\kappa_1$ and $\kappa_2$ in order to find the values of $\kappa_1$ and $\kappa_2$ that minimize the expected number of queries under a uniform distribution assumption. The second and third experiments use the values of $\kappa_1$ and $\kappa_2$ found on the first experiment and compare the minmax ITP method with the relaxed version where $N_{\max} = N_{1/2}+1$ and interpolation search over both artificial and real data.

\paragraph{Artificial Data 1} In our first experiment, we search for the values of $\kappa_1$ and $\kappa_2$ that minimize the expected number of iterations required by the minmax ITP method over lists of size $n = 2\times 10^5$. As seen in the proof of Theorem \ref{the:ITP_charaterization}, the behaviour of the ITP method quickly mimics the behaviour of $k_t$ which  depends solely on the values of $\kappa_1$ and $\kappa_2$. We performed $10^4$ Monte Carlo simulations by generating the list $\boldsymbol{v}$ by sorting $n$ independent samples from a uniform distribution over $[0,1]$. The target value $z$ was also sampled from a uniform distribution over $[0,1]$. Table \ref{tab:training_results} shows the empirical average obtained by varying $\kappa_1$ between $0.01$ and $0.78$, and, varying $\kappa_2$ between $0.51$ and $0.99$. 

\begin{table}[ht]
\caption{\label{tab:training_results} Average number of iterations of $10^4$ Monte Carlo simulations of searches in lists of size $n = 2\times 10^{5}$ and $z$ sampled from a uniform distribution between $0$ and $1$. Each column shows the performance of the ITP method with a given value of $\kappa_1$ and each line a fixed value of $\kappa_2$.}
\centering
\begin{tabular}{llcccccccc}
& & \multicolumn{8}{c}{$\kappa_1:$}\\
& & $0.01$  & $0.12$  & $0.23$  & $0.34$  & $0.45$  & $0.56$  & $0.67$  & $0.78$  \ignore{& $0.89$  & $1.00$  } \\
\cmidrule(lr){3-10}
& $0.51$ & \LL{7.44}  & 7.35  & 7.68  & 8.02  & 8.40   & 8.76  & 9.09  & 9.38   \ignore{& 9.67 & 9.94}   \\ 
& $0.56$ & \LL{7.39} &   7.37  & 7.79 &  8.30 &  8.85 &  9.35  & 9.78  & 10.16   \ignore{& 10.50  & 10.80 } \\
& $0.62$ & \LL{7.31} &  7.43 &  8.08  & 8.87 &  9.57 &  10.17  & 10.66  & 11.07  \ignore{ & 11.47 & 11.78}  \\
& $0.67$ & \LL{7.20}  & 7.63 & 8.66  & 9.63  & 10.46  & 11.13 &  11.68  & 12.13 \ignore{ &  12.54 & 12.88}\\ 
$\kappa_2:$& $0.72$ & \LL{7.08} &  8.05 &  9.38 &  10.54 & 11.51  & 12.25  & 12.81  & 13.32 \ignore{&  13.72   & 14.06}\\
& $0.78$ & \LL{6.95} &  8.64 &  10.26 &  11.66  & 12.73  & 13.52  & 14.09 &  14.51 \ignore{ &  14.86 &15.13}\\
& $0.83$ & \LL{\cellcolor[gray]{0.9} 6.87} & 9.35 &  11.40 &  13.00  & 14.09  & 14.73  & 15.19 & 15.57 \ignore{ & 15.88 & 16.05} \\
& $0.88$ &  \LL{6.93}  & 10.27  & 12.78  & 14.42 &  15.23 &  15.60 &  15.93 &  16.30  \ignore{ & 16.68 & 16.97}\\
&$0.94$ & \LL{7.21} &  11.45 &  14.39 &  14.95  & 15.94 &  16.77  & 17.24  & 17.51  \ignore{  & 17.65 & 17.69}\\
&$0.99$ & \LL{7.51}  & 13.03 &  14.55 &  16.53 &  17.49 &  17.69 &  17.69 &  17.69 \ignore{  & 17.69 & 17.69}
\end{tabular}
\end{table}

As can be seen in Table \ref{tab:training_results}, the empirical average was minimized at $\kappa_1 = 0.01$ and $\kappa_2=0.83$. We highlighted the cell located on the first column and on the seventh row to show the number of iterations attained with these values of $\kappa_1$ and $\kappa_2$ which are significantly lower than $N_{1/2} = \lceil \log_2 2\times 10^5 \rceil = 18$. It should also be noted that the average number of iterations remains below $N_{1/2}$ for any value of $\kappa_1$ and $\kappa_2$ as predicted by Theorem \ref{the:global_minmax}. 

\paragraph{Artificial Data 2}
In our second experiment we compare the empirical performance of two versions of the ITP method against interpolation search on lists of various sizes. The first version of the ITP method used is the minmax version analysed in Theorem \ref{the:ITP_charaterization} and the second version is the one that makes use of the relaxation $N_{\max} = N_{1/2}+1$. The average number of iterations required by each method was calculated by averaging the results of 500 Monte Carlo simulations on lists of sizes ranging from $n = 1$ to $n = 2^{18}$ generated by sorting $n$ independent samples from predetermined distributions. The maximum number of iterations required by interpolation search is also reported for comparison with the worst case performance of the ITP method. In Figure \ref{fig:comparison_ITP_others} we plot the number of iterations as a function of the size of the list for lists generated from the uniform distribution and Figure \ref{fig:comparison_ITP_other_again} show the results under different distributions, namely: when the elements of $\boldsymbol{v}$ are samples of (i) a Gaussian distribution, (ii) an exponential distribution, (iii) a triangular distribution and (iv) a step function distribution (two overlapped uniform distributions over different intervals). The Gaussian in (i) was generated in each run with a random mean $\mu$ sampled from a uniform distribution over $[0,1]$ and a fixed variance with $\sigma = 0.01$; the exponential in (ii) was constructed with a parameter of $\lambda = 1$; the triangular distribution was obtained by taking the square root of a uniformly distributed variable; and the step function distribution in (iv) was obtained by sampling from a distribution that is uniform over the intervals $A = [0, 0.75)$ and $B = [0.75, 1)$ with interval $A$ concentrating half of the cumulative probability and $B$ the other half.

\begin{figure}
	\centering
	\includegraphics[angle = -90, width=90mm]{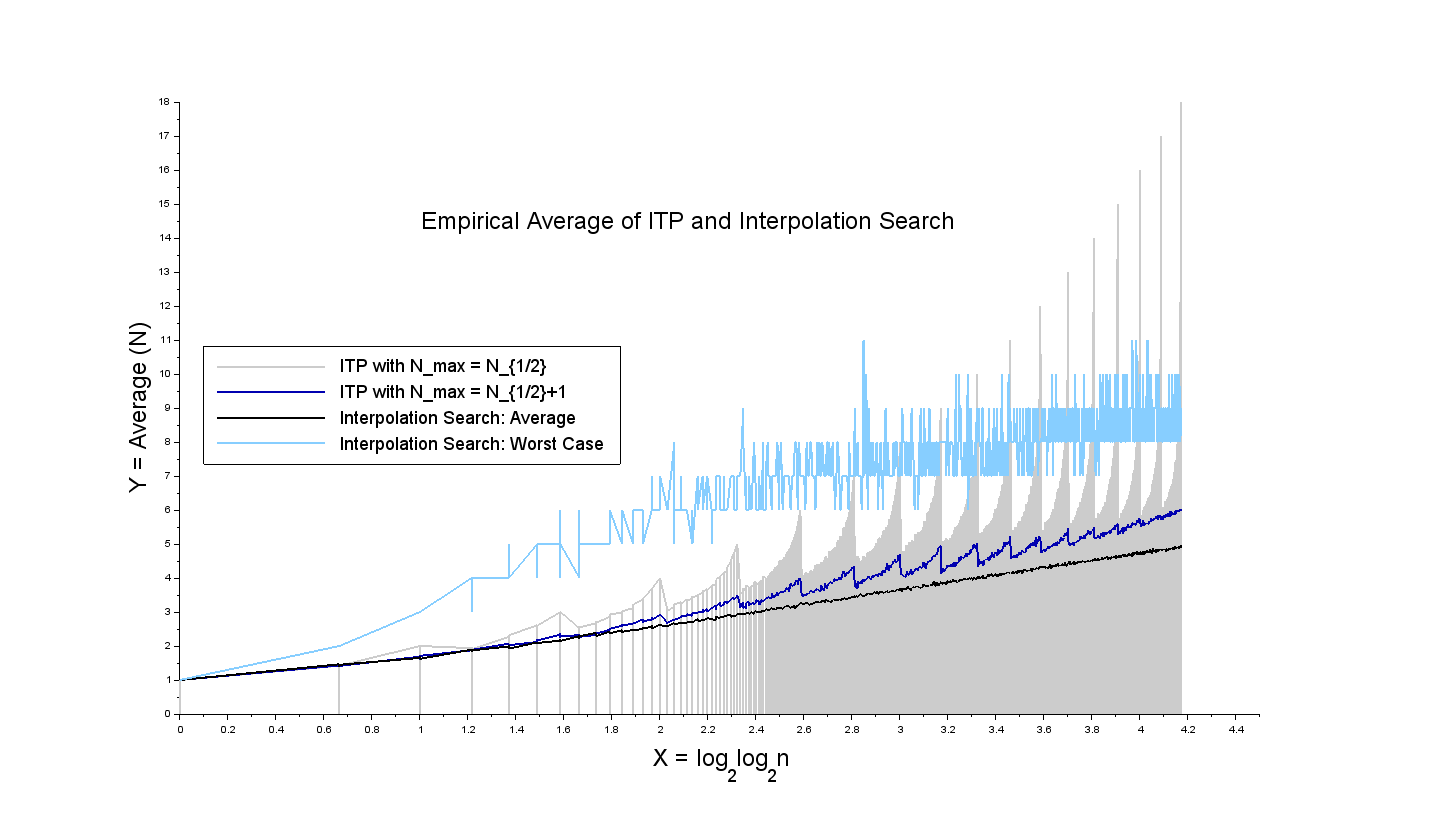}
	\caption{\label{fig:comparison_ITP_others} The average of 500 Monte Carlo simulations comparing two versions of the ITP method and interpolation search  for increasing values of $n$ on data with uniform distribution. In the background, the bar plot in gray displays the average number of iterations required by the minmax version of the ITP method. The lower curve in black shows the average number of iterations required by interpolation search; the dark blue curve above it the average number of iterations required by the ITP method with $N_{\max} = N_{1/2}+1$; and, the light blue curve shows the maximum number of iterations used by interpolation search over all 500 runs.  }
\end{figure} 

In Figure \ref{fig:comparison_ITP_others}, the background bar plot shows the behaviour of the minmax version of the ITP method. As predicted by Theorems \ref{the:global_minmax} and \ref{the:ITP_charaterization}, for values of $n$ in which (\ref{eq:minmax_range}) is not too small, i.e. most of the range, the growth of $E(N)$ is linear with respect to $\log_2 \log_2 n$ similar to interpolation search. The bar plot shows eighteen peaks which correspond to the values of $n$ that are equal to $2^T$ for some $T \in \mathbb{N}$; and thus, for those values of $n$ the number of iterations is linear with $\log_2 n$ and not $\log_2 \log_2 n$. The relaxation of the minmax ITP method with $N_{\max} = N_{1/2}+1$ displayed in dark blue reduces the peaks and obtains a curve that grows linearly with $\log_2 \log_2 n$ in it's entirety just as interpolation search displayed right bellow it in black. The average performance of the ITP method with $N_{\max} = N_{1/2}+1$ when compared with interpolation search, attains an almost identical linear growth with respect to $\log_2 \log_2 n$ that exceeds the number of iterations required by interpolation search by approximately one iteration throughout the range investigated, i.e. the ITP method with  $N_{\max} = N_{1/2}+1$ has a nearly identical expected query complexity as interpolation search under the uniform distribution hypothesis. However, the worst case behaviour of interpolation search is upper-bounded by $n$, i.e. both versions of the ITP method depicted have much better worst case guarantees than interpolation search. The light blue curve overarching the graph depicts the maximum number of iterations required by interpolation search in the 500 runs; and, as can be noticed it exceeded $\log_2 n$ for values of $n$ less than or equal to $2^{2^{3.3}}$ which is approximately $10^3$. Of course, with more runs, interpolation search will demand much more iterations in the worst case.

 \begin{figure}
	\centering
	\includegraphics[angle = -90, width=90mm]{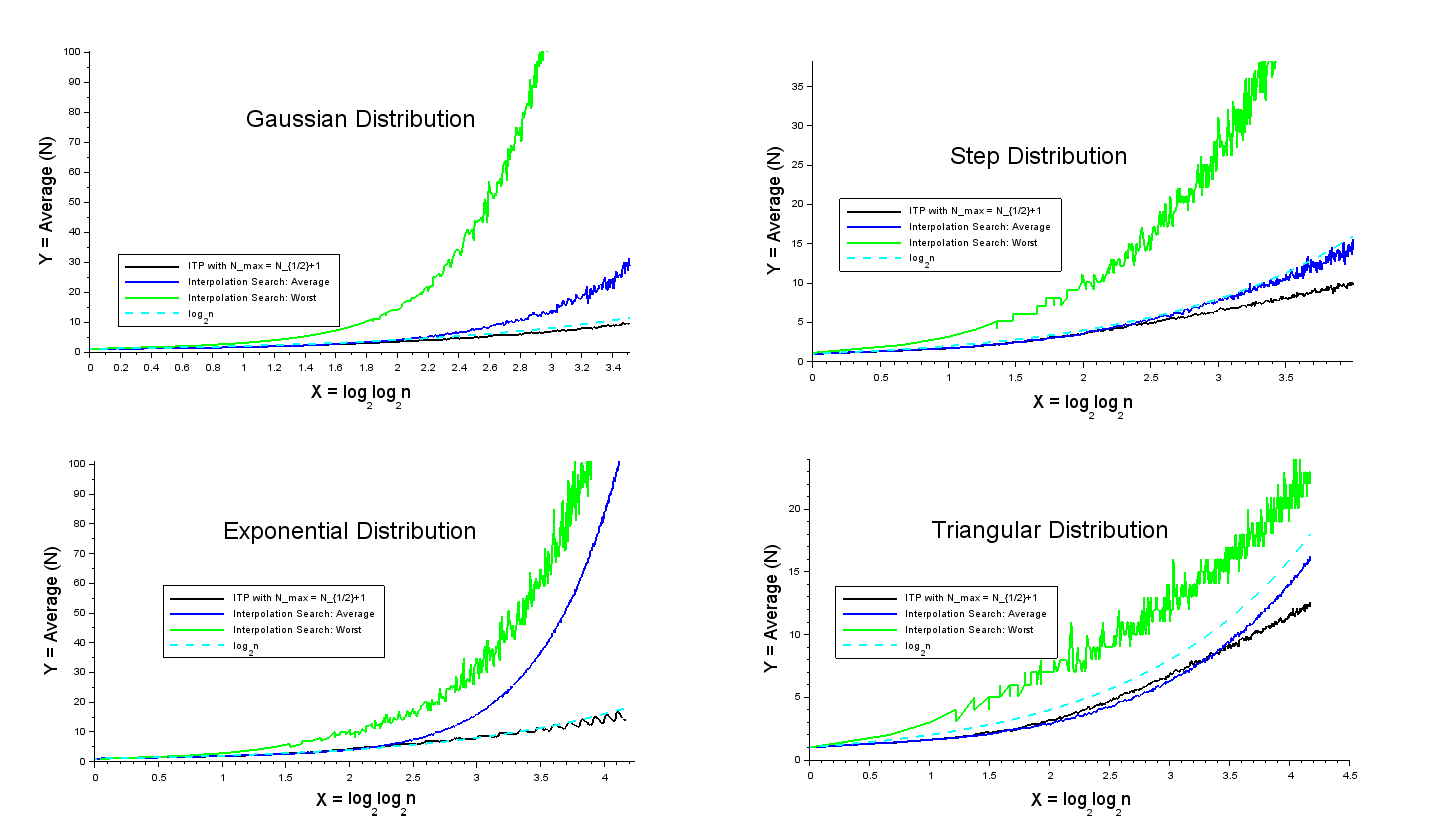}
	\caption{\label{fig:comparison_ITP_other_again} The average of 500 Monte Carlo simulations comparing two versions of the ITP method and interpolation search  for increasing values of $n$. The light blue dashed line provides for reference the value of $\log_2 n$. The lower curve in black shows the average number of iterations required by the ITP method with $N_{\max} = N_{1/2}+1$; the dark blue curve the average number of iterations required by interpolation search; and, the green curve shows the maximum number of iterations used by interpolation search over all 500 runs. }
\end{figure} 

When different distributions are considered then the robustness of the ITP method becomes an interesting feature. As can be seen in Figure \ref{fig:comparison_ITP_other_again}, the average number of iterations of the ITP method with $N_{\max} = N_{1/2}+1$ remained below $\log_2 n$ under all distributions considered. Interpolation search performed much worse than $\log_2 n$ for both the Gaussian distribution and the exponential distribution, and displayed an average performance that seems to be close to $\log_2 n$ under the step function and the triangular distribution considered. The worst case behaviour of interpolation search showed to be much worse than $\log_2 n$ under the four distinct distributions. As depicted in Figure \ref{fig:comparison_ITP_other_again}, under these distributions and others still, interpolation search may have both an average and a worst case performance that require much more iterations than the ITP method by several orders of magnitude. Thus, these experiments show that the ITP method seems to be a much better alternative than both binary searching and interpolation searching when both worst case and average performances are taken into account.

\paragraph{Real Data} In our third experiment we collect a wide variety of real data from publicly available lists of varying sizes and different origins which are specified in the appendix section. To name a few, we have included a list of full names of all public employees of the Brazilian government, a dataset of genome sizes of fungal species, atomic weights, zip codes and others. For each list we calculate the empirical average of the number of iterations required by both the ITP method with $N_{\max} = N_{1/2}+1$ and  interpolation search. In each run we sample $z$ between $v_0$ and $v_n$ with a uniform probability and perform the search with both methods. Four of the twelve lists considered were composed of names rather then numbers, specifically the NASDAQ Acronyms, the English Dictionary, the Family Names and Full Names. These  were converted into numerical lists by taking a base-27 read of each digit and sorting them accordingly. Other natural  approaches that could be used  are the ASCII standard conversion or even a Morse code mapping onto binary numbers. Clearly, the average performance of the ITP method is sensitive to this mapping and hence there is space for improvement. However, this goes beyond the scope of the paper and thus we opted to display only the results for the first approach considered, i.e. the base 27 conversion. Table \ref{tab:real_data} reports the empirical average of $10^3$ runs of the described procedures.

\begin{table}[ht]
\caption{\label{tab:real_data} Average number of iterations required by the ITP method and Interpolation Search. The averages are taken over $10^3$ searches for a target $z$ sampled from a uniform distribution between $v_0$ and $v_n$. We also report the empirical maximum number of iterations required by Interpolation Search over this sample. The simulation capped out the count when more than $1000$ iterations were required, and thus when $1000$ iterations are reached we indicate with a sub-index the number of runs where this cap occurred. Also, since  the ITP method was implemented with $N_{\max}=N_{1/2}+1$,  in the second column under the title ``ITP Search'' we provide the value of $N_{1/2}+1$.  On the bottom lines, the estimates of the mean and the median are displayed in units of $N_{1/2}$.}
\centering
\begin{tabular}{l cc cc cc}
&  \multicolumn{2}{c}{ITP Search} & \multicolumn{2}{c}{Interp Search}\\ \cmidrule(lr){2-3}  \cmidrule(lr){4-5}
&mean & max & mean & max & $N_{1/2}$ & $n$  \\
\cmidrule(lr){1-7}
Thermodynamics Table &  3.6  & 7 &   3.1 &  5 & 6 & 49  \\
Atomic Weights &  3.3 & 7 & 2.8 & 7 & 6 & 54  \\
Fluid Dynamics Chart &  5.8 & 11 &  5.8 & 11 & 10 & 600   \\
Fibonacci Sequence &   8.2 &  11 & 19.8 & 553 & 10 & 700   \\
Genome Sizes &  9.6 & 13 & 8.9 & 27 & 12  & 2352  \\
NASDAQ Acronyms & 10.6 & 15 &  28.8 & 1000$_{18}$ & 14 & 8203  \\
Zip Codes &  10.5 & 18 & 9.8 & 69 & 17 & 81831\\
Family Names & 16.5 &  18 & 90.0 & 1000$_{14}$ & 17 & 88799  \\
English Dictionary & 19.0 & 20 & 247.4 & 1000$_{103}$ & 19 & 370103 \\
Full Names  & 20.6 & 21 & 751.3 & 1000$_{667}$ & 20 & 660276\\
Prime Numbers & 7.2 & 21 & 6.0 &  10 & 20 & 664579  \\
Harmonic Series &  22.3 &  25 & 79.7 & 189 & 24 & $10^7$   \\
\cmidrule(lr){1-7}
\multicolumn{2}{c}{Central Tendency Metrics:}\\
\cmidrule(lr){1-7}
Median & \multicolumn{2}{c}{ $ 0.75 N_{1/2}$} &  \multicolumn{2}{c}{ $ 5.53 N_{1/2}$}  \\
Median  & \multicolumn{2}{c}{ $0.78 N_{1/2}$}   & \multicolumn{2}{c}{ $1.36 N_{1/2}$} 
\end{tabular}

\end{table}

Table \ref{tab:real_data} displays the average number of iterations required by the ITP method side by side with the number of iterations required by interpolation search. The ITP method seems to have a better performance when compared to interpolation search under both the average query complexity criteria and the worst case query complexity criteria. In all instances where interpolation search outperformed the ITP method on the average, it did so by less than $1.21$ iterations, and when the ITP method outperformed interpolation search it did so by up to $730.75$ iterations which is more than 36 times the number of of iterations required by the ITP method. On the average the ITP method required $25\%$ less iterations when compared to binary searching whereas interpolation search required on average more than five times the number of iterations as binary searching across all twelve lists. We point out that even if outliers were excluded from the list (the two most difficult cases for interpolation search for example) interpolation search still attains an empirical average worse than binary search, i.e. interpolation search does not seem to perform well in real data. One possible explanation for this might be the fact that real world data is not generated from uniformly distributed variables, and hence, the robustness guarantees provided by the ITP method seem to be vital for outperforming binary search in real world applications. By analysing the median metric a similar conclusion is reached, i.e. interpolation search performs poorly and the ITP method outperforms binary search.

When considering the worst case performances, since the ITP method in display  made use of the relaxation $N_{\max} = N_{1/2}+1$, then the ITP method never required more than one iteration above $N_{1/2}$, but due to the $O(\log_2 \log_2 n)$ expected query complexity, under favorable conditions it performed less than half the number of iterations of binary searching. On the other hand, interpolation search not only averaged higher iteration counts but it also maxed out the number of iterations with several unsuccessful searches, and hence, it seems to be the least interesting alternative amongst the three when both metrics are taken into consideration.

\paragraph{General Recommendations}   Throughout our experiments (including an extensive number of experiments not reported here) the performance of the ITP method with the relaxation $N_{\max} > N_{1/2}$ seems to give the best results overall. With the relaxation, the ITP method is less sensitive to the value of $n$ but also less sensitive to the choice of $\kappa_1$ and $\kappa_2$. As a rule of thumb we recommend the ITP method with $\kappa_1=0.01$ and $\kappa_2=0.83$ and with the relaxation\footnote{By adopting a non integer value for $N_{\max}$, the maximum number of iterations of Algorithm 0 is of  $\lceil N_{\max} \rceil$. Furthermore, the projection step of the ITP method projects to the interior of $\mathcal{I}_j$ instead of the border, avoiding numerical errors associated to edge cases. Our experiments were performed with $N_{\max} = N_{1/2}+1$, however for practitioners we recommend a non integer $N_{\max}$ such as $N_{\max} = N_{1/2}+0.99$ instead.} of $N_{\max} = N_{1/2}+0.99$, however, if prior knowledge on the distribution over instances of (\ref{def.problem}) is available, or, if there is availability of a training set, then the values of $\kappa_1, \kappa_2$ and $N_{\max}$ can be tested and chosen accordingly. 

In Experiment 2, both interpolation search and the ITP method were implemented under four misspecified conditions when the non-uniform distributions were used to generate $\boldsymbol{v}$. If prior knowledge of the underlying distribution is available, then the behaviour of Algorithm $0$ depicted in Figure \ref{fig:comparison_ITP_others} can be obtained for different distributions by implementing Algorithm $0$ on the transformation of $v_j$ by the cumulative distribution.


 
\section{Discussion}\label{sec:discussions}

In this paper we have identified a novel and yet simple searching method, which we refer to as the ITP method, that attains an expected query complexity of $O(\log_2 \log_2 n)$ iterations and a worst case query complexity of $\lceil \log_2 n \rceil$; i.e. it is optimal with respect to both average and worst case metrics.  Furthermore, we also prove robustness guarantees that show that binary search cannot outperform the ITP method by more than a constant factor even if the distributional hypothesis is misspecified. Hence, the ITP method enjoys the benefits of interpolation search (the improved expected query complexity of $\log_2 \log_2 n$) without the drawbacks associated with it (a lower than binary search expected query complexity when distribution is misspecified). We perform extensive testing on artificial and real data and we find that the ITP method can considerably outperform both the classical binary search method as well as interpolation search. We reach time-savings that range roughly from 25\% to 75\%, depending on the experiment, when compared to binary searching; an overall much better performance than interpolation search when compared across experiments.

Binary searching is a fundamental tool in the field of computer science and has continually been the choice for applications, specifically due to its minmax optimality. Our results show that this preference for binary search, or alternatively for interpolation search, has often been an inefficient one.  The improvements highlighted here have both practical and theoretical implications that directly translate to significant time savings, specifically when the cost of a query is much greater than the time to compute the procedure itself. In short, the ITP method is our recommended improvement to the traditional approach. However, the identified minmax class of methods, which is largely unexplored, is potentially a more significant contribution that may lead to further improvements and the identification of even more efficient methods.

\paragraph{Future work} The problem of searching sorted tables and/or other multidimensional variants are natural instances that may benefit if equivalent results as the ones developed here are found. Another relatively unexplored variation studied in \cite{bentley} is searching through infinite lists. Also, assuming multiple instances of (\ref{def.problem}) to be solved sequentially and generated under one common distribution,  one may ask how to  adapt and improve the solvers in between each resolution to obtain an adaptive/self-improving method.  Finally, the cost of one query is typically assumed to be significantly greater than that of the computation of the searching procedure itself; several interesting questions arise when this assumption is modified.

\appendix

\section{Online material}

Table \ref{tab:data_scource} contains the sources of the twelve lists used in the second experiment. The texts were converted into numerals as explained in the end of Section \ref{sec:experiments} and any additional symbols such as ``*.!;" and others were ignored.

\begin{table}[ht]
\caption{\label{tab:data_scource} The source of the data collected on the fifth of June of 2019.}
\centering
\begin{tabular}{l}
\cmidrule(lr){1-1} 
NASDAQ Acronyms \\ \footnotesize{\url{ftp://ftp.nasdaqtrader.com/symboldirectory}}   \\
Prime Numbers \\ \footnotesize{{\color{brown}(self generated)}}    \\
Atomic Weights \\ \footnotesize{\url{https://www.qmul.ac.uk/sbcs/iupac/AtWt/}}   \\
Zip Codes \\ \footnotesize{\url{http://federalgovernmentzipcodes.us}} \\
Fluid Dynamics Chart \\ \footnotesize{\url{https://engineering.purdue.edu/~wassgren/notes/CompressibleFlowTables.xls}}   \\
Genome Sizes \\ \footnotesize{\url{http://www.zbi.ee/fungal-genomesize/index.php?q}}    \\
Fibonacci Sequence \\ \footnotesize{{\color{brown}(self generated)}}     \\
Thermodynamics Table \\ \footnotesize{\url{https://www.ohio.edu/mechanical/thermo/property_tables/H2O/H2O_TempSat.xls}}    \\
English Dictionary \\ \footnotesize{\url{https://github.com/dwyl/english-words}}  \\
Family Names \\ \footnotesize{\url{https://www.census.gov/topics/population/genealogy/data/2010_surnames.html}}    \\
Harmonic Series \\ \footnotesize{{\color{brown}(self generated)}}   \\
Full Names  \\ \footnotesize{\url{http://www.portaltransparencia.gov.br/servidores}}  \\
\cmidrule(lr){1-1} 
\end{tabular}
\end{table}
Several of the files found in the above links contain multiple columns, specifically the fluid dynamics chart, the genome sizes, the atomic weights and the thermodynamics table. When this is the case  we selected one column arbitrarily and performed all simulations on the selected column.

\bibliography{references}

\end{document}